\documentclass{article}

\usepackage{amsthm} 
\usepackage{amsmath} 
\usepackage{amssymb}
\usepackage[scanall]{psfrag}
\usepackage{graphics}
\usepackage{graphicx}

\theoremstyle{plain}
\newtheorem{thm}{Theorem}[section]
\newtheorem*{thm*}{Theorem}

\theoremstyle{definition}

\newtheorem*{defn*}{Definition}

\begin{document}

\author{ Johanna Becker\thanks{MTA-ELTE Egerv\'ary Research Group
    (EGRES), Institute of Mathematics, E\"otv\"os University,
    Budapest, P\'azm\'any P.~s.~1/C, Hungary H-1117. Research is
    supported by France Telecom R \& D, by OTKA grants K60802,
    TS049788 and by European MCRTN Adonet, Contract Grant No.~504438.
    e-mail: {\tt \{beckerjc, csisza, jacint, szego\}@cs.elte.hu}} ,
    Zsolt Csizmadia$^*$,\\ Alexandre Laugier\thanks{France Telecom
    Research \& Development / BIZZ / DIAM.  e-mail: {\tt
    Alexandre.Laugier@francetelecom.com}}, J\'acint Szab\'o$^*$,
    L\'aszl\'o Szeg\H{o}$^*$}

\title{Balancing congestion for unsplittable routing on a bidirected ring}

\maketitle

\begin{abstract}
  Given a bidirected ring with capacities and a demand graph, we
  present an approximation algorithm to the problem of finding the
  minimum $\alpha$ such that there exists a feasible unsplittable
  routing of the demands after multiplying each capacity by
  $\alpha$. We also give an approximation scheme to the problem.
\end{abstract}

\section{Introduction}
In this paper a bidirected ring means the union of two oppositely
directed circuits on the same set of nodes.  The motivation of the
present paper is to compute a routing of demands in an SDH
(Synchronous Digital Hierarchy) based network. Such a network is a
transmission network and we will refer to it as the supply graph.
Also, we will refer to the set of demands as the demand graph. The
backbone networks of the European telecommunication companies are
based on the SDH technology. The ring architecture plays a key role in
this technology, because it provides an efficient self-healing
mechanism in case of failure: although the rings are bidirected, only
one direction is used, while the other would be useful in case of a
failure. Most of the SDH backbone networks consist of such
uni-directed circuits. From such uni-directed circuits different
topologies have been built, like tree of circuits or cycle of
circuits. In the next section we will show that the routing in a
supply graph which is a cycle of circuits can be reduced to a routing
problem in a bidirected ring.  One can find more about the SDH
transmission protocol in \cite{tanenbaum}.

Formally, we are given a directed supply graph $G$ with capacity
function $c\colon E(G)\rightarrow \mathbb{R}_+$ and a directed demand
graph $H$ on the same node set with demand values $d\colon
E(H)\rightarrow \mathbb{R}_+$.  A \textbf{routing} of $H$ is a
collection of $uv$-paths of value $k$, one for each demand edge of
$E(H)$ joining $u$ to $v$ with value $k$, satisfying the capacity
constraint $c$. Since each demand is routed along a unique path we
will speak about \textbf{unsplittable} routings. In this paper a
routing is always unsplittable, unless when speaking about the
fractional solution of the linear relaxation of the problem.  We say
that a routing of $H$ has \textbf{load at most (less than) $c$} if for
each edge $e\in E(G)$ the sum of the values of the paths using $e$ is
at most (less than) $c(e)$.

In this paper we consider the routing problem in bidirected rings.
Given a bidirected ring $G$ with capacity function
$c\,\colon\,E(G)\rightarrow \mathbb{R}_+$, and a directed demand graph
$H$ with demand values $d\colon E(H)\rightarrow \mathbb{R}_+$, we want
to find the minimum $\alpha$ for which there exists a routing of $H$
with load at most $\alpha c$. This is called the \textsc{balanced
bidirected ring routing problem}. Under balancedness we mean that the
value of $\alpha$ is introduced to ensure that the remaining
capacities are as large as possible, allowing the network to carry
larger demands. The case when $c$ is uniform (that is, all capacities
are equal) and each demand is 1 was solved with an elegant method by
Wilfong and Winkler \cite{ww}.  In Section \ref{approx} we follow the
lines of their proof to give an algorithm which returns a routing with
load less than $\alpha_{opt} c+\frac{3}{2} D$ where $\alpha_{opt}$ is
the optimum solution of the problem and $D$ is the maximum value of
the demands.  We also show that this error term can indeed occur.

In Section \ref{khanna}, using a method of Khanna \cite{khanna}, we
show an approximation scheme to the problem which for any $\varepsilon
> 0$ yields a routing with load less than $\alpha_{opt}(c+\varepsilon
\frac{\sum c}{n})$, with a trade-off with the running time.  Here
$n=|V(G)|$ and $\sum c$ is the sum of the capacities of the edges.

By solving a linear program, a fractional solution of the
\textsc{balanced bidirected ring routing problem} can be given, that
is, we can calculate the minimum $\alpha$, denoted by $\alpha^*$, such
that there exists a fractional routing of $H$ with load at most
$\alpha c$. Clearly $\alpha^*\leq \alpha_{opt}$. Our approximation
algorithm for the \textsc{balanced bidirected ring routing problem}
gives a routing such that its load is actually less than $\alpha^*
c+\frac{3}{2}D$. If $\alpha^*<1$ then our solution may be considered
good in practice since it requires less than $\frac{3}{2}D$ additional
capacity on each edge. On the other hand, if $\alpha^*\geq 1$ then
there exists no feasible routing with the given capacities. Hence we
are facing a network design problem, and our goal is to increase the
capacity function with minimum cost, in order to have a routing
satisfying the increased capacity.  We show in Section \ref{design}
that even this problem can be handled with techniques similar to
Section \ref{approx}.

If the capacity function $c$ is uniform and the ring is
\emph{undirected}, the balanced ring routing problem was studied by
many researchers due to its significance in telecommunication
networks. This problem is called the \textsc{undirected ring loading
problem}, first considered by Cosares and Saniee
\cite{cosares}. 
Schrijver,
Seymour, and Winkler \cite{seysch} gave a combinatorial approximation
algorithm for the undirected case, if $c$ is not necessarily
uniform. Their algorithm returns a routing requiring less than
$\frac{3}{2}D$ more capacities on each edge than in an optimum
solution. 
The \textsc{balanced bidirected ring
  routing problem} was first considered by Wilfong and Winkler \cite{ww} who
gave an exact algorithm for finding an optimum routing in a bidirected ring,
in case $c$ is uniform and each demand is 1. Our considerations in Section
\ref{approx} are based on their method, yielding a generalization of their
result. Our result are more general than theirs in the aspect that in our case
the demands are not restricted to be 1, and the capacity function is not
necessarily uniform.

We point out that the \textsc{balanced bidirected ring routing
  problem} is NP-complete. Indeed, the \textsc{partition problem} can
be reduced to it in a straightforward way, just as in the undirected
case (Cosares, Saniee \cite{cosares}).  Moreover, contrary to the
undirected case, the cut condition is not sufficient for the existence
of a fractional solution. That may be the reason that no combinatorial
algorithm is known finding an optimum fractional solution of the
bidirected ring routing problem, unlike in the undirected case.

\section{Motivation}
In some SDH backbone networks the supply graph $G$ forms a cycle of
directed circuits, with capacity function $c: E(G)\rightarrow
\mathbb{R}_+$. In this section we show that the routing in $G$ can be
reduced to a routing problem in a bidirected ring.  Two neighboring
circuits of $G$ meet in a node called a \textbf{common node}. The common
nodes divide each directed circuit into two edge sets, called
\textbf{arcs}. See~Figure \ref{red}.

\begin{figure}
\begin{center}
  \includegraphics[scale=.3]{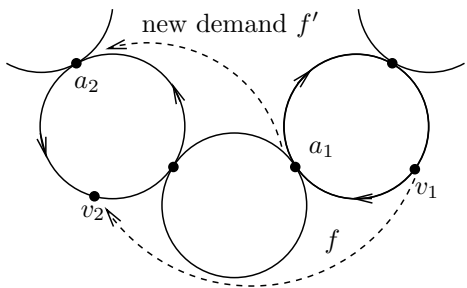}
\end{center}
\caption{Reduction to a bidirected ring}\label{red}
\end{figure}
    
Let $f\in E(H)$ be a demand edge with value $d$ joining $v_1$ to
$v_2$.  Let $a_1$ be the common node $a$ of $G$ minimizing
$dist_G(v_1,a)$. Similarly, let $a_2$ be the common node $a$
minimizing $dist_G(a,v_2)$. If $v_2$ is contained in the $v_1a_1$-path
of $G$ then only one $v_1v_2$-path exists in $G$ so simply delete $f$
from $H$ and decrease the capacity of this path by $d$. Otherwise all
$v_1v_2$-paths use the supply edges on the $v_1a_1$- and on the
$a_2v_2$-paths of $G$. So replace $f$ by a new demand edge $f'$
joining $a_1$ to $a_2$ with value $d$ and decrease the capacities of
the supply edges on the $v_1a_1$ and on the $a_2v_2$-paths by $d$.
After doing so for each demand edge, the new demand graph $H'$ has a
routing in $G$ with load at most the new capacity function $c'$ if and
only if $H$ has a routing with load at most $c$. Moreover, the
non-common nodes of $G$ are incident to no edge of $H'$, so we can
think of an arc $A$ as only one edge with capacity the minimum of
$c'(e)$ taken over all edges $e\in A$. The new supply graph is a
bidirected ring with $n$ nodes and $2n$ edges, where $n$ is the number
of circuits in $G$.

\section{An approximation algorithm}\label{approx}
In this section we show an algorithm yielding a routing in a
bidirected ring with load less than $\alpha_{opt} c+\frac{3}{2}D$
where $\alpha_{opt}$ is the solution of the \textsc{balanced
bidirected ring routing problem} and $D$ is the maximum value of the
demands.  Our algorithm is a modification of that of Wilfong and
Winkler \cite{ww} who solved the \textsc{balanced bidirected ring
routing problem} in the case if $c$ is uniform and the demands are
$1$. We do not count running times since solving a linear program is
included.

From the two directions of the ring we say that one is the
\textbf{forward} and the other one is the \textbf{backward} direction.
Accordingly, an edge $e\in E(G)$ can be forward or backward, and from
the two possible $uv$-paths ($u,v\in V(G)$) one is the forward and the
other one is the backward path. For an edge $f\in E(H)$ joining $u$ to
$v$, an \textbf{$f$-path} refers to any of the two $uv$-paths of $G$.
The edge sets of these two paths will be denoted by $F(f)$ and $B(f)$,
respectively.

The first step is to solve the LP-relaxation of the problem. There are
two possibilities of routing demand $f\in E(H)$ hence we introduce a
variable $0\leq \varphi(f)\leq 1$ with the meaning that $\varphi(f)$
fraction of the demand $f$ is routed forward and $1-\varphi(f)$
fraction is routed backward. So we have the following LP-relaxation,
whose optimum solution is denoted by $\alpha^*$. Both sums run on
demand edges $f\in E(H)$.

\[\min\,\, \alpha, \mbox{ s.t.}\]
\begin{equation}\label{LP}
0\leq \varphi\leq 1 
\end{equation}
\[\sum_{f \,:\, e\in F(f)} \varphi(f)\, d(f)+\sum_{f \,:\, e\in B(f)} (1-\varphi(f))\,
d(f)\leq \alpha c(e)\,\,\,\,\,\,\,\,\forall e\in E(G).\]

Note that we would get $\alpha_{opt}$ if $\varphi$ was required to be
integer. Now we manipulate the demands $f\in E(H)$ with
$0<\varphi(f)<1$.  Such demands are called \textbf{split}.  We say that
demands $f_1,\,f_2\in E(H)$, where $f_i$ joins $s_i$ to $t_i$ for
$i=1,2$, are \textbf{parallel} if the end nodes are placed in the ring
in the order $s_1,\,t_1,\,t_2,\,s_2$ (some of these nodes may
coincide).  Assume that $f_1,\, f_2\in E(H)$ are parallel split
demands. Call the $f_i$-path containing both $s_{3-i}$ and $t_{3-i}$
the \textbf{long $f_i$-path} for $i=1,2$. Let $x_i$ denote the amount of
flow of $f_i$ sent along the long $f_i$-path in our fractional
solution.  If we reroute $\min(x_1,x_2)$ amount of flow from the long
$f_i$-path to the other $f_i$-path for $i=1,2$ then one of the demands
$f_1,\,f_2$ will not be split any more, moreover, we do not increase
the load of any edge of $G$ (we may even decrease it somewhere). So at
most $|E(H)|$ such \textbf{uncrossing steps} are possible and finally we
get a fractional solution where there are no pair of parallel split
demands.  Especially, demands with the same source node are
parallel. So it will hold that for each $s\in V(G)$ there exists at
most one split demand $f\in E(H)$ with source $s$.

Denote the nodes of $G$ by $s_1,\,\ldots,\,s_n$ in the forward order.
Now we try to unsplit the remaining split demands. Let $f_i$ be the
split demand with source $s_i$ (if any). Assume that $x_i$ fraction of
$f_i$ is routed forward and $y_i$ fraction backward.  Let $w_i=y_i$ if
we would set $\varphi(f_i)$ to 1 and $w_i=-x_i$ if $\varphi(f_i)$
would be set to 0.  If we round $\varphi(f_i)$ to 0 or 1 then the load
of an edge $e\in E(G)$ increases by $w_i$ if $e$ is contained in the
forward $f_i$-path, it decreases by $w_i$ if $e$ is contained in the
backward $f_i$-path, and it does not change elsewhere. There are no
two parallel split demands so the change of the load of an edge is
$\pm \sum_{j\leq i\leq k}w_i$ for some $j\leq k$ where $k$ may be
greater than $n$ but then the indices of $w$ are meant modulo $n$.
Here the sign depends on whether $e$ is a forward or a backward edge.
Now we try to set the value of $w_i$ to $y_i$ or $-x_i$ for all $i$ in
such a way that $|\sum_{j\leq i\leq k}w_i|<\frac{3}{2}D$ holds for all
$j\leq k$.  To achieve this it is clearly sufficient that
\begin{equation}\label{D2}
-D/2<\sum_{1\leq i\leq k}w_i\leq D/2 \,\, \mbox{ holds for all } k\leq n.
\end{equation}
(\ref{D2}) can be easily achieved by greedily setting the values $w_i$
to $y_i$ or $-x_i$ one after another, since $x_i+y_i\leq D$ holds.
Hence by this procedure each load is increased by \emph{less than}
$\frac{3}{2}D$.

\bigskip So the algorithm for the \textsc{balanced bidirected ring routing
  problem} is the following.

\begin{enumerate}
\item Solve the LP-relaxation of the \textsc{balanced bidirected ring
    routing problem}.
\item Uncross the parallel split demands.
\item Unsplit the remaining split demands in the above described
  greedy way.
\end{enumerate}

\begin{thm}\label{pol}
  The algorithm gives a routing of $H$ with load less than $\alpha^*
  c+\frac{3}{2}D\leq \alpha_{opt}c+\frac{3}{2}D$.
\end{thm}

\begin{proof}
After the uncrossing procedure the routing still has load at most
$\alpha^* c$. Moreover, as we observed, the load of each edge $e\in
E(G)$ increases by less than $\frac{3}{2}D$ during unsplitting.
\end{proof}

Our goal was to find a routing with load at most
$\alpha_{opt}c+\frac{3}{2}D$, but actually, similarly to the
undirected case \cite{seysch}, the routing we got has load less than
$\alpha^* c+\frac{3}{2}D\leq \alpha_{opt}c+\frac{3}{2}D$. However,
this fact does not always induce improved efficiency. Indeed, next we
show that for every $\delta>0$ there is an example where the output
routing of the algorithm has load not less than $\alpha_{opt}
c+\frac{3}{2}D - \delta$. Such \emph{tight examples} are very
important in understanding how algorithms can be improved in order to
decrease running time.

\begin{figure}
\begin{center}
  \includegraphics[scale=0.35]{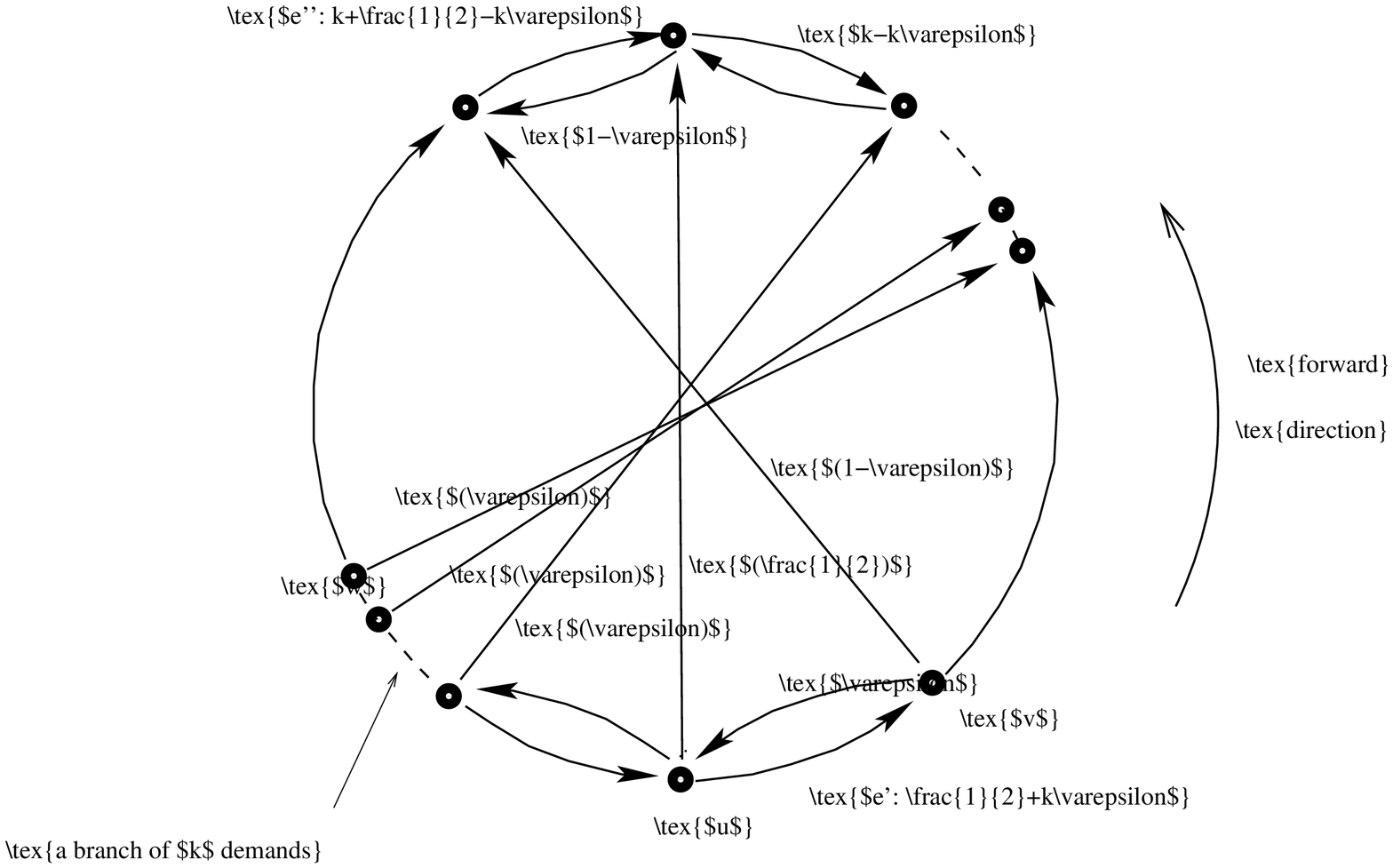}
\end{center}
\caption{An optimum fractional routing}\label{SharpGraph}
\end{figure}

Consider the capacities of the ring as shown in Figure
\ref{SharpGraph}. The edges without capacity, as well as the dashed
edges have large capacities in both directions and $\varepsilon$ is a
number which is small enough with respect to $k$.  All demands are of
value $1$. Observe that the cut $\{e',\,e''\}$ is tight, that is, the
total value of the demands crossing it equals $c(e')+c(e'')$.  It is
easy to see that there is no integer routing with $\alpha=1$ due to
the cut $\{e',\,e''\}$ with fractional capacities.  The best integer
solution is obtained if the capacity of $e''$ is raised to $k+1$, that
is $\alpha_{opt}^{(k)} = \frac{k+1}{k+\frac{1}{2}-k\varepsilon}$.
Indeed, $\varepsilon$ is small enough so we can route the demand of
node $v$ forward while all other demands backward. Thus $\lim_{k\to
\infty} \alpha_{opt}^{(k)}=1$ by an appropriate choice of
$\varepsilon=\varepsilon^{(k)}$.

Clearly $\alpha^*=1$, and an optimum $\varphi$ (which is a vertex of
the polyhedron (\ref{LP})) is shown in brackets in Figure
\ref{SharpGraph}. Assume that the unsplitting procedure starts at node
$u$. The split demand at node $u$ is routed forward.  Now the split
demand at node $v$ is routed backward and then the demand of node $w$
forward. All remaining demands will be routed backward, according to
our greedy heuristic. Now $e'$ has load $2$ and so the error at $e'$
tends to
\[
\lim_{k\to \infty}\frac{2-(\frac{1}{2}+k\varepsilon)\cdot
  \alpha_{opt}^{(k)}}{D} = \frac{3}{2},
\]
by choosing $k\varepsilon \to 0$, and using that $D=1$. Note that the
fractional capacities can be made integer by scaling.


\section{An approximation scheme}\label{khanna}
If the value of the demands are small with respect to the minimum load
of an optimum routing then the $\frac{3}{2}D$ additive error is a
small deviation from the optimum solution. Otherwise this error term
can be significant.  In this section we develop an approximation
scheme for the \textsc{balanced bidirected ring routing problem},
which for any $\varepsilon > 0$ gives an algorithm polynomial in the
number of nodes, finding a routing with load less than
$\alpha_{opt}(c+\varepsilon \overline{c})$ where $n=|V(G)|$, and
$\overline{c}=\frac{\sum_{e\in E(G)}c(e)}{n}$.  We use the solution
method of Khanna \cite{khanna} who presented an
$\varepsilon$-approximation scheme for the case of an undirected ring
with $c$ uniform. By definition, an $\varepsilon$-approximation scheme
for a minimization problem is an algorithm which, for any $\varepsilon
> 0$, returns a solution with value at most $1+\varepsilon$ times the
optimum value. In this sense the method of this section is not an
$\varepsilon$-approximation scheme since it does not approximate
$\alpha_{opt}$. However, for $M>0$, in the class of \textsc{balanced
bidirected ring routing problem} instances, where $\overline{c}\leq
Mc(e)$ holds for all edge $e$, our method is an
$M\varepsilon$-approximation scheme. Many problem instances in
practice arising from an SDH backbone network belong to such a class
with reasonably small $M$, in particular the instances where $c$ is
uniform.

For $s,t\in V(G)$ the \textbf{long $st$-path} is defined to be the
longest $st$-path in $G$, or to be the forward one if the two paths are
of equal length.  Recall that $\alpha ^*$ denotes the fractional
optimum to the \textsc{balanced bidirected ring routing problem}.  Let
$\alpha'\geq \alpha^*$ be any real number and let
\begin{center}
\begin{math}
E'=\left\{f\in E(H)\,:\, d(f) >
  \frac{2}{3}\varepsilon\alpha'\overline{c}\right\}.
\end{math}
\end{center}
Note that 
\begin{center}
\begin{math}
\alpha' \sum_{e\in E(G)} c(e) \geq \alpha^* \sum_{e\in E(G)} c(e) \geq
\sum_{f\in E'}d(f) > |E'|\frac{2}{3}\varepsilon\alpha'\overline{c}.
\end{math}
\end{center}
 Hence $|E'|< \frac{3}{2\varepsilon}n$. If $\alpha'\geq \alpha_{opt}$
 then less than $\frac{3}{\varepsilon}$ edges of $E'$ are routed in
 the long path in any optimum routing, since otherwise the sum of the
 loads would be more than
\begin{center}
\begin{math}
\frac{2}{3}\varepsilon\alpha'\overline{c} \cdot \frac{3}{\varepsilon}
\cdot \frac{n}{2}= n\alpha'\overline{c} =\alpha'\sum_{e\in E(G)} c(e)
\geq \alpha_{opt}\sum_{e\in E(G)} c(e),
\end{math}
\end{center}
which is impossible. Thus, independently of whether $\alpha'\geq
\alpha_{opt}$ or not, for all subsets $E'' \subseteq E'$ with $|E''| <
\frac{3}{\varepsilon}$ we do the following. We route the demands of
$E''$ in the long paths and the demands of $E'- E''$ in the short
paths.  Denote the load of $e\in E(G)$ in this routing of $E'$ by
$l(e)$. Now denote by $\alpha^*_{E''}$ the optimum of the following
linear program, where the sums run on $f\in E(H)- E'$.

\[\min\,\, \alpha, \mbox{ s.t.}\]
\[0\leq \varphi\leq 1\]
\[\sum_{e\in F(f)} \varphi(f)\, d(f)+\sum_{e\in B(f)} (1-\varphi(f))\,
d(f) + l(e)\leq \alpha c(e)\,\,\,\,\,\,\,\,\forall e\in E(G).\]

Note that the maximum value of a demand in $E(H)-E'$ is at most
$\frac{2}{3}\varepsilon\alpha'\overline{c}$. Hence exactly as in the
previous section, we can find a routing of $E(H)-E'$ with load less
than $\alpha^*_{E''} c-l+\varepsilon\alpha'\overline{c}$.  In the case
when $\alpha'\geq \alpha_{opt}$, in any optimum routing of $H$ less
than $\frac{3}{\varepsilon}$ edges of $E'$ are routed in the long
path, hence we get that one of the above routings has load less than
$\alpha_{opt} c+\varepsilon\alpha'\overline{c}$. For any $\alpha'$,
from these $\sum_{i=0}^{3/\varepsilon} \dbinom{|E'|}{i}$ routings of
$E(H)$ choose the one with load at most $\alpha c+\varepsilon \alpha'
\overline{c}$, such that this $\alpha$ is minimum.

The number of subsets $E''\subseteq E'$ to consider is
$\sum_{i=0}^{3/\varepsilon} \dbinom{|E'|}{i}$.  Using that
$\dbinom{k}{l} \leq (ek)^l/l^l$ holds for any integers $k\geq l$, when
$|E'|\geq \frac{6}{\varepsilon}$ we get
\[\sum_{i=0}^{3/\varepsilon} \dbinom{|E'|}{i}  
\leq \frac{3}{\varepsilon}
\left(\frac{e|E'|}{3/\varepsilon}\right)^{\frac{3}{\varepsilon}} <
\frac{3}{\varepsilon}
\left(\frac{en}{2}\right)^{\frac{3}{\varepsilon}}.\] Since the left
hand side is monotone increasing in $|E'|$, the above bound is valid
when $|E'|<\frac{6}{\varepsilon}$, too. Hence the number of subsets
$E''\subseteq E'$ to try is
$O_\varepsilon(n^{\frac{3}{\varepsilon}})$.

Now we show how $\alpha'$ can be chosen. It is clear that
$\alpha^*\leq \alpha_{opt}\leq 2\alpha^*$, since there are exactly two
paths between any two nodes in $G$. So first determine the value of
$\alpha^*$ and then run the above algorithm with
$\alpha'=\alpha_i=\frac{N+i}{N}\alpha^*$ for $i=0,\ldots,N$ for some
integer $N$. As we mentioned, at the point when $\alpha_{opt}\leq
\alpha_j\leq (1+1/N)\alpha_{opt}$ happens to hold, our routing has
load less than $\alpha_{opt}( c+\varepsilon(1+1/N)\overline{c})$.
Thus, finally, from these $N+1$ routings choose the one with load at
most $\alpha(c+\varepsilon(1+1/N) \overline{c})$, such that this
$\alpha$ is minimum. In the beginning we should also replace
$\varepsilon$ by $\varepsilon/(1+1/N)$ in the above scheme.

\section{Network design}\label{design}
The approximation algorithm presented in Section \ref{approx} gives a
routing in a bidirected ring with load less than
$\alpha^*c+\frac{3}{2}D$. Recall that we can calculate the fractional
optimum $\alpha^*$ by solving a linear program. Observe that
$\alpha^*<1$ means that there exists a fractional routing in the
supply network $G$ with the given capacities. Since our approximate
integer solution requires less than $\frac{3}{2}D$ additional capacity
on each edge (hopefully, much less on average), this solution may be
considered good in practice. On the other hand, if $\alpha^*\geq 1$
then there exists no feasible routing with the given capacities. Hence
we are facing a network design problem, where we want to increase the
capacity of some edges of the ring, with minimum cost, in order to
guarantee the existence of a routing satisfying the increased
capacities.

Formally, a \textbf{widening cost} $w_e$ is given on each edge $e\in
E(G)$, measuring the cost of increasing the capacity of edge $e$ by
one unit. We aim to find the minimum of $\sum_{e\in E(G)}\gamma_ew_e$
such that there exists an integer routing with load at most
$(\gamma_e+c(e))(1-\alpha)$ on any edge $e\in E(G)$. Here $\alpha$ is
some robustness factor known a priori. Next we give a heuristic
algorithm for this problem, with similar methods than that of Section
\ref{approx}. First we solve the fractional relaxation of the problem.

\[\min\,\, \sum_{e\in
  E(G)}\gamma_ew_e, \mbox{ s.t.}\]
\begin{equation}
0\leq \varphi\leq 1 
\end{equation}
\[\sum_{e\in F(f)} \varphi(f)\, d(f)+\sum_{e\in B(f)} (1-\varphi(f))\,
d(f)\leq (\gamma_e+c(e))(1-\alpha)\,\,\,\,\,\,\,\,\forall e\in E(G).\]

Next uncross the split demands, and then unsplit the remaining split
demands, just as in Section \ref{approx}. Exactly as in Section
\ref{approx}, we get an integer routing of $H$ with load less than
$(\gamma^*_e+c(e))(1-\alpha)+\frac{3}{2}D$, with the optimum
$\gamma_e^*$.

Hence we have a routing with widening cost which is at most
$\frac{3D}{2}\sum_{e\in E(G)} w_e$ more expensive than the optimum
cost. We can also choose $\alpha$ to ensure that the load of this
routing is at most $\gamma_e^*+c(e)$, yielding a feasible unsplit
routing in the new, increased network.

We mention that it is also possible to apply the method of Section
\ref{khanna} to yield for any $\varepsilon > 0$ an algorithm
polynomial in $n$, finding a routing with widening cost at most
$\varepsilon\frac{C^*}{n} \sum_{e\in E(G)} w_e$ more expensive than
the optimum cost, where $C^*$ is the minimum sum of loads of
fractional routings of $H$ in the uncapacitated ring $G$ (one should
simply define $E'=\left\{f\in E(H)\,:\, d(f) > \frac{2}{3}\varepsilon
  \frac{C'}{n}\right\},$ where $C^* \leq C' \leq 2C^*$).

\end{document}